\documentclass[journal]{IEEEtranTIE}
\usepackage{graphicx}
\pdfoutput=1
\usepackage{cite}
\usepackage{picinpar}
\usepackage{amsmath}
\usepackage{url}
\usepackage{flushend}
\usepackage[latin1]{inputenc}
\usepackage{colortbl}
\usepackage{soul}
\usepackage{multirow}
\usepackage{pifont}
\usepackage{color}
\usepackage{alltt}
\usepackage[hidelinks]{hyperref}
\usepackage{enumerate}
\usepackage{siunitx}
\usepackage{breakurl}
\usepackage{epstopdf}
\usepackage{pbox}
\usepackage{subfigure}
\usepackage{amsfonts}

\newtheorem{lemma}{Lemma}
\newtheorem{remark}{Remark}

\newtheorem{proposition}{Proposition}
\newenvironment{proof}{{\indent \indent \it Proof:}}

\begin{document}
\title{	An Adaptive Multivariable Smooth Second-Order Sliding Mode Approach}

\author{
	\vskip 1em	
	Xidong Wang
	\thanks{
Xidong Wang is with the Research Institute of Intelligent Control and Systems, School of Astronautics, Harbin Institute of Technology, Harbin 150001, China (e-mail: 17b904039@stu.hit.edu.cn).
	}
}

\maketitle
	
\begin{abstract}
This paper presents a novel adaptive multivariable smooth second-order sliding mode approach with the features of fast finite-time convergence, adaptation to disturbances and smooth. This approach can be directly applied to the controller design of multi-input and multi-output (MIMO) systems. In addition, a novel adaptive multivariable smooth disturbance observer is proposed based on this structure. In terms of the types of disturbances, the fast finite-time convergence and the fast finite-time uniformly ultimately boundedness of the systems are proved with the corresponding fast finite-time Lyapunov stability theory. Finally, the effectiveness of the proposed approach is validated by comparative numerical simulations.
\end{abstract}

\begin{IEEEkeywords}
Adaptive multivariable smooth second-order sliding mode control (AMSSOSMC), Adaptive multivariable smooth disturbance observer (AMSDO), Fast finite-time convergence, Fast finite-time uniformly ultimately boundedness.
\end{IEEEkeywords}

%\markboth{IEEE TRANSACTIONS ON INDUSTRIAL ELECTRONICS}%
{}

\definecolor{limegreen}{rgb}{0.2, 0.8, 0.2}
\definecolor{forestgreen}{rgb}{0.13, 0.55, 0.13}
\definecolor{greenhtml}{rgb}{0.0, 0.5, 0.0}

\section{Introduction}

\IEEEPARstart{I}{n} recent years, sliding mode control has attracted much attention because of its insensitivity and strong robustness to the parameter and external disturbance uncertainty. However, the chattering phenomenon existing in traditional sliding mode control restricts its application in practice. To diminish the chattering effect, the concept of high order sliding mode (HOSM) is proposed. In the present HOSM algorithms, the super-twisting algorithm is very popular and own practical application value due to the characteristics of finite-time convergence, strong robustness and solely requiring the information of sliding mode variables \cite{Levant2003}. Numerous modified super-twisting methods have been proposed to further improve the control performance \cite{2007Smooth,Moreno2008,Shtessel2012,ASTC2015,Jiang2018,2020arXiv}.

However, the above-mentioned control methods can only be utilized in single-input single-output (SISO) system. When applied in the MIMO system, the system needs to be decoupled into multi-SISO systems, which constrains the application of super-twisting sliding mode control. In \cite{2014Multi}, a multivariable super-twisting sliding mode approach is present to directly design the controller for MIMO systems, which obtains better control performance than that of the single super-twisting sliding mode control. In \cite{2015Adaptive_Multi}, an adaptive multivariable super-twisting sliding mode control method is proposed, which can adapt to the disturbance of unknown boundary.

In this paper, inspired by [9], we extend our previous work [7] to multivariate form and propose a new proposition called the adaptive multivariable smooth second-order sliding mode control (AMSSOSMC) approach, which integrates the merits of fast finite-time convergence, adaptation to the disturbances and smooth. This new method can be directly applied to the controller design of MIMO systems. Moreover, a novel adaptive multivariable smooth disturbance observer (AMSDO) is also proposed based on this new method, which can adjust the parameters automatically without a priori knowledge of the upper bound of disturbances derivative and have smoother output than that of adaptive multivariable disturbance observer (AMDO) proposed in [9].

In terms of the types of disturbances, the fast finite-time convergence and the fast finite-time uniformly ultimately boundedness of the systems will be proved with the corresponding finite-time Lyapunov stability theory. The effectiveness and superiority of the proposed method is verified by comparative simulation experiments.

The rest of this paper is organized as follows. In Section II, some necessary lemmas are given. The adaptive multivariable smooth second-order sliding mode approach is provided in Section III. Contrastive numerical simulations are executed to verify the effectiveness of the proposed approach in Section V. Section VI concludes this paper.

Notation: In this paper, we use $\left\| \cdot \right\|$  for the Euclidean norm of vectors and $\otimes$ for the Kronecker product. ${\lambda _{\max }}\left(  \cdot  \right)$ and ${\lambda _{\min }}\left(  \cdot  \right)$ denote the maximum and minimum eigenvalues of a matrix, respectively. ${I_{nn}}$ and ${I_{n}}$ denote the $n \times n$ identity matrix and $n$ dimension unit column vector, respectively.  ${0_{nn}}$ and ${0_{n}}$ denote $n \times n$ zero matrix and $n$ dimension zero column vector, respectively. 
%2-A
\section{Preliminaries}
To better describe the following \emph{Lemmas}, we consider a general system
\begin{equation}
\dot x = h(x(t)),{x_0} = x(0)
\end{equation}
where $h:{U_1} \to {R^n}$  is continuous on an open neighborhood ${U_1} \subset {R^n}$ of the origin and assume that $h(0) = 0$. The solution of (5) is denoted as $x(t,{x_0})$, which is understood in the sense of Filippov \cite{1999Filippov}.

\begin{lemma} [fast finite-time stability \cite{2005lemma1}]
Suppose there exists a continuous and positive-definite function $V:{U_1} \to R$ such that the following condition holds:
\begin{equation}
\dot V(x) \le  - {c_1}V{(x)^p} - {c_2}V(x)
\end{equation}
where ${c_1} > 0,{c_2} > 0,p \in (0,1)$, then the trajectory of (5) is fast finite-time stable, and the settling time is given by:
\begin{equation}
T \le \frac{{\ln [1 + {c_2}V{{({x_0})}^{1 - p}}/{c_1}]}}{{{c_2}(1 - p)}}
\end{equation}
\end{lemma}

\begin{lemma} [fast finite-time uniformly ultimately boundedness \cite{Jiang2018}]
Suppose there exists a continuous and positive-definite function $V:{U_1} \to R$ such that the following condition holds:
\begin{equation}
\dot V(x) \le  - {c_1}V{(x)^{{p_1}}} - {c_2}V(x){\rm{ + }}{c_3}V{(x)^{{p_2}}}
\end{equation}
where ${c_1} > 0,{c_2} > 0,{c_3} > 0,{p_1} \in (0,1),{p_2} \in (0,{p_1})$, then the trajectory of (5) is fast finite-time uniformly ultimately boundedness, and the settling time is given by:
\begin{equation}
T \le \frac{{\ln [1 + ({c_2} - {\theta _2})V{{({x_0})}^{1 - {p_1}}}/({c_1} - {\theta _1})]}}{{({c_2} - {\theta _2})(1 - {p_1})}}
\end{equation}
where ${\theta _1}$ and ${\theta _2}$ are arbitrary positive constants holding ${\theta _1} \in (0,{c_1}),{\theta _2} \in (0,{c_2})$, then $x(t,{x_0})$ can converge to a region of equilibrium point in a finite time $T$. In addition, the residual set of solution of (5) can be given by:
\begin{equation}
D = \left\{ {x:{\theta _1}V{{(x)}^{{p_1} - {p_2}}} + {\theta _2}V{{(x)}^{1 - {p_2}}} < {c_3}} \right\}
\end{equation}

Define an auxiliary variable ${\theta _3} \in \left( {0,1} \right)$. If ${\theta _3}$ is selected satisfying
\begin{equation}
{\theta _3}^{1 - {p_2}}{\theta _2}^{{p_1} - {p_2}}{c_3}^{1 - {p_2}} = {\theta _1}^{1 - {p_2}}{(1 - {\theta _3})^{{p_1} - {p_2}}}
\end{equation}
then (10) can be reduced to $D = {D_1} = {D_2}$, where
\begin{equation}
\begin{aligned}
{D_1}& = \left\{ {x:V{{(x)}^{{p_1} - {p_2}}} < {\theta _3}{c_3}/{\theta _1}} \right\}\\
{D_2}& = \left\{ {x:V{{(x)}^{1 - {p_2}}} < \left( {1 - {\theta _3}} \right){c_3}/{\theta _2}} \right\}
\end{aligned}
\end{equation}
which means the state $x$ can converge to ${D_1} = {D_2}$ in finite time $T$.
\end{lemma}

%3-A
\section{An Adaptive Multivariable Smooth Super-Twisting Sliding Mode Approach}
Motivated by [9], we extend our previous work [7] to multivariate form and obtain the following new proposition.
\begin{proposition}
Considering the following system
\begin{equation}
\begin{aligned}
{\dot x_1} &=  - {L_1}(t)\frac{{{x_1}}}{{{{\left\| {{x_1}} \right\|}^{{\textstyle{1 \over m}}}}}} - {L_2}(t){x_1} + {x_2}\\
{\dot x_2} &=  - {L_3}(t)\frac{{{x_1}}}{{{{\left\| {{x_1}} \right\|}^{{\textstyle{2 \over m}}}}}} - {L_4}(t){x_1} + d
\end{aligned}
\end{equation}
where ${x_1},{x_2}\in\mathbb{R}^n$ and $\left\| d \right\| \le \delta $, $\delta$ is an unknown non negative constant. The adaptive gains ${L_1}(t),{L_2}(t),{L_3}(t),{L_4}(t)$ are formulated as
\begin{equation}
\begin{aligned}
{L_1}(t)&={k_1}L_0^{{\textstyle{{m - 1} \over m}}}\left( t \right),{\rm{  }}{L_2}(t) = {k_2}L_0\left( t \right)\\
{L_3}(t)& = {k_3}L_0^{{\textstyle{{2m - 2} \over m}}}\left( t \right),{\rm{  }}{L_4}(t) = {k_4}L_0^2\left( t \right)
\end{aligned}
\end{equation}
where ${k_1},{k_2},{k_3},{k_4},m$ are positive constants satisfying
\begin{equation}
{m^2}{k_3}{k_4} > \left( {\frac{{{m^3}{k_3}}}{{m - 1}} + \left( {4{m^2} - 4m + 1} \right){k_1}^2} \right){k_2}^2,m > 2
\end{equation}

${L_0}(t)$ is a positive, time-varying, and scalar function. The ${L_0}(t)$ satisfies
\begin{equation}
\dot{L_0}\left( t \right) =
\begin{cases}
{\kappa},&if \quad \left\| {x_1} \right\| \ge \varepsilon \hfill \\
0,&else \hfill \\
\end{cases} 
\end{equation}
where $\kappa$ is a positive constant, $\varepsilon$ is an arbitrary small positive value. Then the following statements hold

(\romannumeral1) If $d = 0$, then $x_1,\dot x_1$ can fast converge to the origin in finite time.

(\romannumeral2) If $d \ne 0$ and $d$ is a bounded disturbance, then $x_1,\dot x_1$ can fast converge to a region of the origin in finite time.
\end{proposition}

\begin{proof}
To facilitate the analysis, define the following new state vectors
\begin{equation}
\begin{aligned}
{\xi _1} = \frac{{L_0^{\frac{{m - 1}}{m}}}}{{{{\left\| {{x_1}} \right\|}^{\frac{1}{m}}}}}{x_1},{\xi _2} = {L_0}{x_1},{\xi _3} = {x_2}
\end{aligned}
\end{equation} 
where ${\xi _i}\in\mathbb{R}^n (i=1,2,3)$.

Taking the derivative of (13) yields
\begin{equation}
\begin{aligned}
\dot \xi  = \left[ {\begin{array}{*{20}{c}}
{{{\dot \xi }_1}}\\
{{{\dot \xi }_2}}\\
{{{\dot \xi }_3}}
\end{array}} \right] = \left[ {\begin{array}{*{20}{c}}
{\frac{{m - 1}}{m}\frac{{{{\dot L}_0}{x_1}}}{{{{\left[ {{L_0}\left\| {{x_1}} \right\|} \right]}^{\frac{1}{m}}}}} + \frac{{L_0^{\frac{{m - 1}}{m}}}}{{{{\left\| {{x_1}} \right\|}^{\frac{1}{m}}}}}\left[ {{I_{nn}} - \frac{1}{m}\frac{{{x_1}x_1^T}}{{{{\left\| {{x_1}} \right\|}^2}}}} \right]{{\dot x}_1}}\\
{{{\dot L}_0}{x_1} + {L_0}{{\dot x}_1}}\\
{ - {k_3}\frac{{L_0^{\frac{{2m - 2}}{m}}{x_1}}}{{{{\left\| {{x_1}} \right\|}^{\frac{2}{m}}}}} - {k_4}L_0^2{x_1} + d}
\end{array}} \right]
\end{aligned}
\end{equation}

In view of the definition in (13), one has
\begin{equation}
\begin{aligned}
\frac{{{x_1}}}{{\left\| {{x_1}} \right\|}} = \frac{{{\xi _1}}}{{\left\| {{\xi _1}} \right\|}} = \frac{{{\xi _2}}}{{\left\| {{\xi _2}} \right\|}}
\end{aligned}
\end{equation} 

Then
\begin{equation}
\begin{aligned}
&\left[ {{I_{nn}} - \frac{1}{m}\frac{{{x_1}x_1^T}}{{{{\left\| {{x_1}} \right\|}^2}}}} \right]{\dot x_1}\\ 
&= \left[ {{\xi _3} - \frac{{{\xi _1}{\xi _1}^T{\xi _3}}}{{m{{\left\| {{\xi _1}} \right\|}^2}}}} \right] - \sum\limits_{i = 1}^2 {{k_i}} \left[ {{I_{nn}} - \frac{1}{m}\frac{{{\xi _i}{\xi _i}^T}}{{{{\left\| {{\xi _i}} \right\|}^2}}}} \right]{\xi _i}\\ 
&= \left[ {{\xi _3} - \frac{{{\xi _1}{\xi _1}^T{\xi _3}}}{{m{{\left\| {{\xi _1}} \right\|}^2}}}} \right] - \left[ {\frac{{m - 1}}{m}\left( {{k_1}{\xi _1} + {k_2}{\xi _2}} \right)} \right]
\end{aligned}
\end{equation}

Substituting (16) into (14) results in 
\begin{equation}
\begin{aligned}
\dot \xi  =  - \frac{{{L_0}}}{{{{\left\| {{\xi _1}} \right\|}^{\frac{1}{{m - 1}}}}}}A\xi  - {L_0}B\xi  + C
\end{aligned}
\end{equation}
where
\begin{equation}
\begin{aligned}
&A = \left[ {\begin{array}{*{20}{c}}
{\frac{{m - 1}}{m}{k_1}}&{\frac{{m - 1}}{m}{k_2}}&{ - 1}\\
0&0&0\\
{{k_3}}&0&0
\end{array}} \right] \otimes {I_{nn}}\in\mathbb{R}^{3n \times 3n}\\
&B = \left[ {\begin{array}{*{20}{c}}
0&0&0\\
{{k_1}}&{{k_2}}&{ - 1}\\
0&{{k_4}}&0
\end{array}} \right] \otimes {I_{nn}}\in\mathbb{R}^{3n \times 3n}\\
&C = {C_1} + {C_2} + {C_3}\in\mathbb{R}^{3n}
\end{aligned}
\end{equation}
with
\begin{equation}
\begin{aligned}
&{C_1} = \frac{{m - 1}}{m}\frac{{{{\dot L}_0}}}{{{L_0}}}\left[ {\begin{array}{*{20}{c}}
{{\xi _1}}\\
{\frac{m}{{m - 1}}{\xi _2}}\\
{{0_n}}
\end{array}} \right]\\
&{C_2} = \left[ {\begin{array}{*{20}{c}}
{{0_n}}\\
{{0_n}}\\
d
\end{array}} \right]\\
&{C_3} = \left[ {\begin{array}{*{20}{c}}
{\frac{{ - {L_0}{\xi _1}{\xi _1}^T{\xi _3}}}{{m{{\left\| {{\xi _1}} \right\|}^{2 + \frac{1}{{m - 1}}}}}}}\\
{{0_n}}\\
{{0_n}}
\end{array}} \right] 
\end{aligned}
\end{equation}

For the system (9), select a positive definite Lyapunov function as $V\left( \xi  \right) = {\xi ^T}P\xi$, where
\begin{equation}
\begin{aligned}
P = \frac{1}{2}\left[ {\begin{array}{*{20}{c}}
{\frac{{2m}}{{m - 1}}{k_3} + {k_1}^2}&{{k_1}{k_2}}&{ - {k_1}}\\
{{k_1}{k_2}}&{2{k_4} + {k_2}^2}&{ - {k_2}}\\
{ - {k_1}}&{ - {k_2}}&2
\end{array}} \right]\otimes {I_{nn}}\in\mathbb{R}^{3n \times 3n}
\end{aligned}
\end{equation}
where $P$ is symmetric positive definite. 

Taking the derivative of $V\left( \xi  \right)$ along the trajectories of system (9) yields
\begin{equation}
\begin{aligned}
\dot V =  - \frac{{{L_0}}}{{{{\left\| {{\xi _1}} \right\|}^{\frac{1}{{m - 1}}}}}}{\xi ^T}{\Omega _1}\xi  - {L_0}{\xi ^T}{\Omega _2}\xi  + \widetilde V
\end{aligned}	
\end{equation}
where ${\Omega _1} = {A^T}P + PA$, ${\Omega _2} = {B^T}P + PB$ and $\widetilde V = 2{\xi ^T}PC$

Denoting $\widetilde {{V_i}} = 2{\xi ^T}P{C_i}(i = 1,2,3)$, $\widetilde V$ can be rewritten as
\begin{equation}
\begin{aligned}
\widetilde V = \widetilde {{V_1}} + \widetilde {{V_2}} + \widetilde {{V_3}}
\end{aligned}	
\end{equation}

The $\widetilde {{V_1}}$ and $\widetilde {{V_2}}$ in (22) satisfy the following inequalities
\begin{equation}
\begin{aligned}
&\widetilde {{V_1}} = 2{\xi ^T}P{C_1} \le \frac{{m - 1}}{m}\frac{{{{\dot L}_0}}}{{{L_0}}}{\xi ^T}Q\xi \\
&\widetilde {{V_2}} = 2{\xi ^T}P{C_2} \le \sqrt {k_1^2 + k_2^2 + 4} \left\| \xi  \right\|\delta 
\end{aligned}	
\end{equation}
where $Q = diag\left[ {{q_1},{q_2},{q_3}} \right] \otimes {I_{nn}}$ is a diagonal matrix with positive diagonal elements, which are expressed as follows
\begin{equation}
\begin{aligned}
{q_1}& = \frac{{2m}}{{m - 1}}{k_3} + k_1^2 + \frac{{\left( {2m - 1} \right){k_1}{k_2}}}{{2\left( {m - 1} \right)}} + \frac{{{k_1}}}{2}\\
{q_2} &= \frac{m}{{2\left( {m - 1} \right)}}\left( {4{k_4} + 2k_2^2 + {k_2}} \right) + \frac{{\left( {2m - 1} \right){k_1}{k_2}}}{{2\left( {m - 1} \right)}}\\
{q_3} &= \frac{{{k_1}}}{2} + \frac{{m{k_2}}}{{2\left( {m - 1} \right)}} 
\end{aligned}
\end{equation}

In addition, it can be observed that
\begin{equation}
\begin{aligned}
\frac{{{\xi _3}^T{\xi _1}{\xi _1}^T{\xi _3}}}{{{{\left\| {{\xi _1}} \right\|}^2}}} = \frac{{{{\left( {{\xi _1}^T{\xi _3}} \right)}^2}}}{{{{\left\| {{\xi _1}} \right\|}^2}}} \le {\left\| {{\xi _3}} \right\|^2} 
\end{aligned}
\end{equation}

Thus, the $\widetilde {{V_3}}$  satisfies the following inequality
\begin{equation}\small
\begin{aligned}
&\widetilde {{V_3}} = 2{\xi ^T}P{C_3} \\
&\le  - \frac{1}{m}\frac{{{L_0}}}{{{{\left\| {{\xi _1}} \right\|}^{\frac{1}{{m - 1}}}}}}\left[ {\left( {\frac{{2m{k_3}}}{{m - 1}} + k_1^2} \right)\xi _1^T{\xi _3} - {k_1}{{\left\| {{\xi _3}} \right\|}^2}} \right] - \frac{{{L_0}{k_1}{k_2}}}{m}\xi _1^T{\xi _3}
\end{aligned}
\end{equation}

Incorporating (26) into (21), (21) can be rewritten as
\begin{equation}
\begin{aligned}
\dot V \le  - \frac{{{L_0}}}{{{{\left\| {{\xi _1}} \right\|}^{\frac{1}{{m - 1}}}}}}{\xi ^T}\widetilde{\Omega _1}\xi  - {L_0}{\xi ^T}\widetilde{\Omega _2}\xi  + \widetilde {{V_1}} + \widetilde {{V_2}}
\end{aligned}	
\end{equation}
where
\begin{equation}\small
\begin{aligned}
&\widetilde{\Omega _1}=\\
 &\frac{{{k_1}}}{m}\left[ {\begin{array}{*{20}{c}}
{{k_3}m + k_1^2\left( {m - 1} \right)}&0&{ - {k_1}\left( {m - 1} \right)}\\
0&{{k_4}m + k_2^2\left( {3m - 1} \right)}&{ - {k_2}\left( {2m - 1} \right)}\\
{ - {k_1}\left( {m - 1} \right)}&{ - {k_2}\left( {2m - 1} \right)}&{m - 1}
\end{array}} \right] \otimes {I_{nn}}\\
&\widetilde{\Omega _2} = {k_2}\left[ {\begin{array}{*{20}{c}}
{{k_3} + k_1^2\left( {3m - 2} \right)/m}&0&0\\
0&{{k_4} + k_2^2}&{ - {k_2}}\\
0&{ - {k_2}}&1
\end{array}} \right]\otimes {I_{nn}}
\end{aligned}
\end{equation}

It is easy to prove that the matrices $\widetilde{\Omega _1}$ and $\widetilde{\Omega _2}$ both are positive definite with (11). By using
\begin{equation}
{\lambda _{\min }}\left( P \right){\left\| \xi  \right\|^2} \le V \le {\lambda _{\max }}\left( P \right){\left\| \xi  \right\|^2}
\end{equation}
(27) can be further rewritten as
\begin{equation}\small
\dot V \le  - {L_0}\left( t \right){n_1}{V^{{p_1}}} + {n_2}{V^{\frac{1}{2}}} - \left( {L_0\left( t \right){n_3} - \frac{{2m - 2}}{m}{n_4}\frac{{{{\dot L}_0}}}{{{L_0}}}} \right)V
\end{equation}
where ${p_1} = \left( {2m - 3} \right)/\left( {2m - 2} \right)$,
\begin{equation}
\begin{aligned}
{n_1} &= \frac{{{\lambda _{\min }}\left(\widetilde{{\Omega _1}} \right)}}{{\lambda _{\max }^{{p_1}}\left( P \right)}}\\
{n_2} &= \frac{{\delta \sqrt {k_1^2 + k_2^2 + 4} }}{{\lambda _{\min }^{{\raise0.5ex\hbox{$\scriptstyle 1$}
\kern-0.1em/\kern-0.15em
\lower0.25ex\hbox{$\scriptstyle 2$}}}\left( P \right)}}\\
{n_3} &= \frac{{{\lambda _{\min }}\left(\widetilde {{\Omega _2}} \right)}}{{{\lambda _{\max }}\left( P \right)}}\\
{n_4} &= \frac{{{\lambda _{\max }}\left( Q \right)}}{{2{\lambda _{\min }}\left( P \right)}}
\end{aligned}
\end{equation}

(\romannumeral1) If $d(t)= 0$, then $\delta =0$, (30) will become
\begin{equation}
\dot V \le  - {L_0}\left( t \right){n_1}{V^{{p_1}}} - \left( {L_0\left( t \right){n_3} - \frac{{2m - 2}}{m}{n_4}\frac{{{{\dot L}_0}}}{{{L_0}}}} \right)V
\end{equation}

Due to ${\dot L_0}\left( t \right) \ge 0$, $L_0\left( t \right){n_3} - \left( {2m - 2} \right){n_4}{\dot L_0}/\left( {{L_0}m} \right)$ is positive in finite time. It follows
from (32) that 
\begin{equation}
\dot V \le  - {c_1}{V^{{p_1}}} - {c_2}V
\end{equation}
where ${c_1}$ and ${c_2}$ are positive constants, ${p_1} \in (0.5,1)$. By using \emph{Lemma 1}, $\xi$ can converge to origin in fast finite time, then $x_1,\dot x_1$ can fast converge to the origin in finite time and the proof of (\romannumeral1) is completed.

(\romannumeral2) If $d\left( t \right) \ne 0$, with the same analysis of (\romannumeral1), it follows from (30) that
\begin{equation}
\dot V \le  - {c_4}{V^{{p_1}}} - {c_5}V{\rm{ + }}{c_3}{V^{{\textstyle{1 \over 2}}}}
\end{equation}
where ${c_3},{c_4}$and ${c_5}$ are positive constants, ${p_1} \in (0.5,1)$. By using \emph{Lemma 2}, $\xi$ can converge to a region of origin in fast finite time. In addition, using (9), (10) and (13), the $x_1$ and $\dot x_1$ converge to a region of the origin in fast finite time. The proof of (\romannumeral2) is completed.
 
\end{proof}

\begin{remark}
By adpoting different values of $m$ in (9), a series of adaptive multivariable smooth second-order sliding mode control methods can be obtained. The proposed methods will be used in the design of controller and observer for MIMO systems and the superiority of the proposed control methods will be validated in the next section.
\end{remark}

%4-section
\section{Numerical Simulations}
To facilitate validation of the effectiveness of the proposed approach, the following simplified multivariable perturbed control system is considered
\begin{equation}
{\dot x_1} = u + {d_1}
\end{equation}
Where states ${x_1} = {\left[{x_{11}}\quad {x_{12}}\quad{x_{13}} \right]^T}\in \mathbb{R}^3$, $u \in \mathbb{R}^3$  and disturbance ${d_1} = {\left[{d_{11}}\quad {d_{12}}\quad{d_{13}} \right]^T}\in \mathbb{R}^3$ satisfies $\left\| {{{\dot d}_1}} \right\| \le {\delta _1}$ where ${\delta _1}$ is an unknown non-negative constant. The initial values of the states are set as ${x_1}\left( 0 \right) = {\left[1\quad 3\quad 2 \right]^T}$  
%4-A
\subsection{Experiment \uppercase\expandafter{\romannumeral1}: Comparison with constant disturbance}
In experiment \uppercase\expandafter{\romannumeral1}, the constant disturbance is set as ${d_1} = {\left[0.1\quad 0.2\quad 0.2 \right]^T}$. Based on \emph{Proposition 1}, the controller is expressed as follows 
\begin{equation}\small
\begin{aligned}
u =  - {L_{c1}}(t){\kern 1pt} \frac{{{x_1}}}{{{{\left\| {{x_1}} \right\|}^{{\textstyle{1 \over m}}}}}} &- {L_{c2}}(t){x_1}\\
&- \int_0^t {\left[ {{L_{c3}}(t)\frac{{{x_1}}}{{{{\left\| {{x_1}} \right\|}^{{\textstyle{2 \over m}}}}}} + {L_{c4}}(t){x_1}} \right]d\tau } 
\end{aligned}
\end{equation}
where ${L_{c1}}(t),{L_{c2}}(t),{L_{c3}}(t),{L_{c4}}(t)$ are formulated the same as (15) and the parameters are set as: $m = 3,{k_1} = 2,{k_2} = 2.5,{k_3} = 4,{k_4} = 30,\kappa  = 10$. 

The adaptive multivariable super-twisting sliding mode control (AMSTSMC) method proposed in [9] is implemented as the comparison method, which adopts the same  parameters except that $m$ is set as 2.
\begin{figure}
	\centering
	\subfigure[The curves of ${x_{11}}$ with constant disturbance]{
	\includegraphics[width=0.2\textwidth]{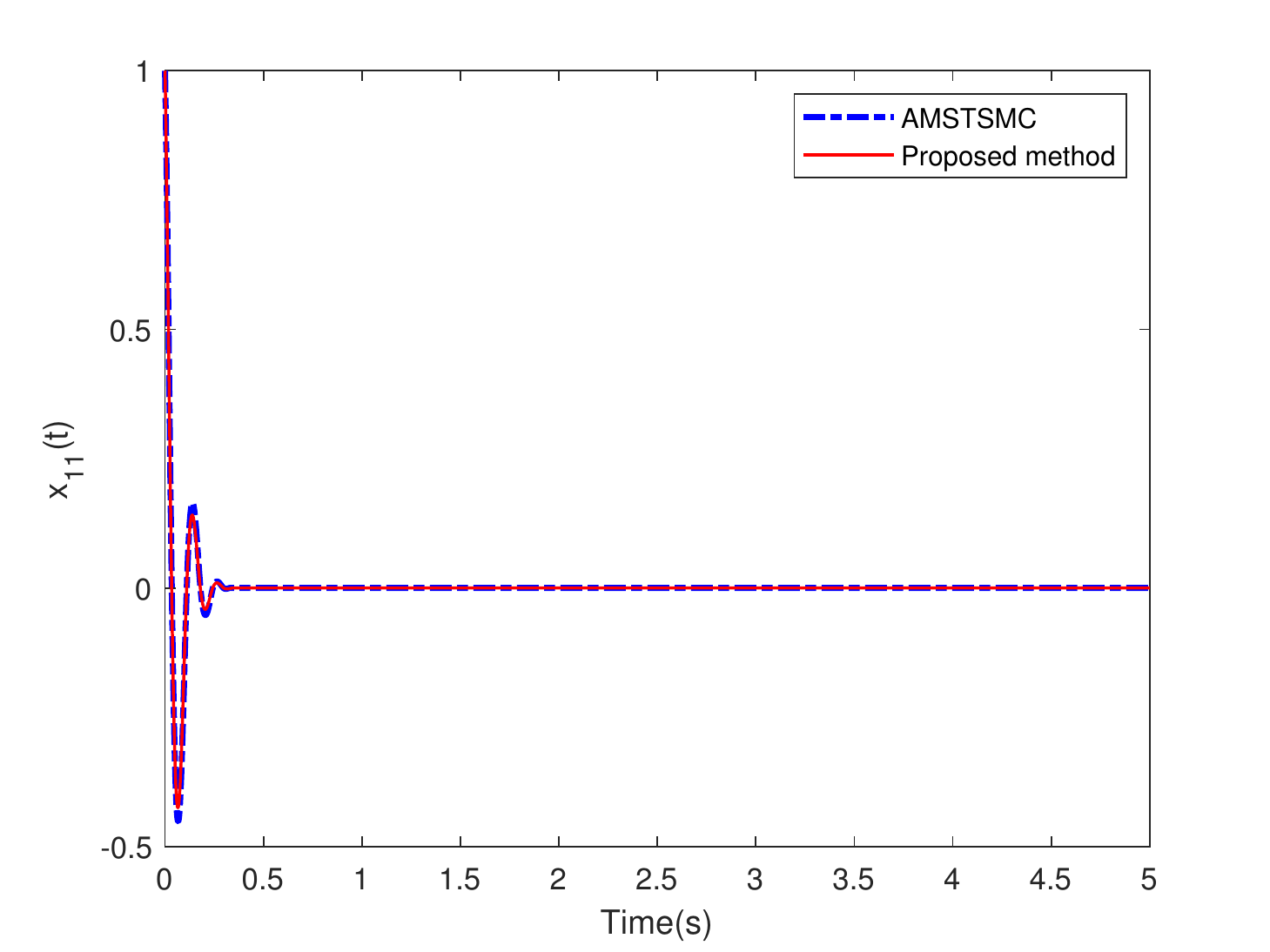}}
	\subfigure[Local magnification of ${x_{11}}$]{
	\includegraphics[width=0.2\textwidth]{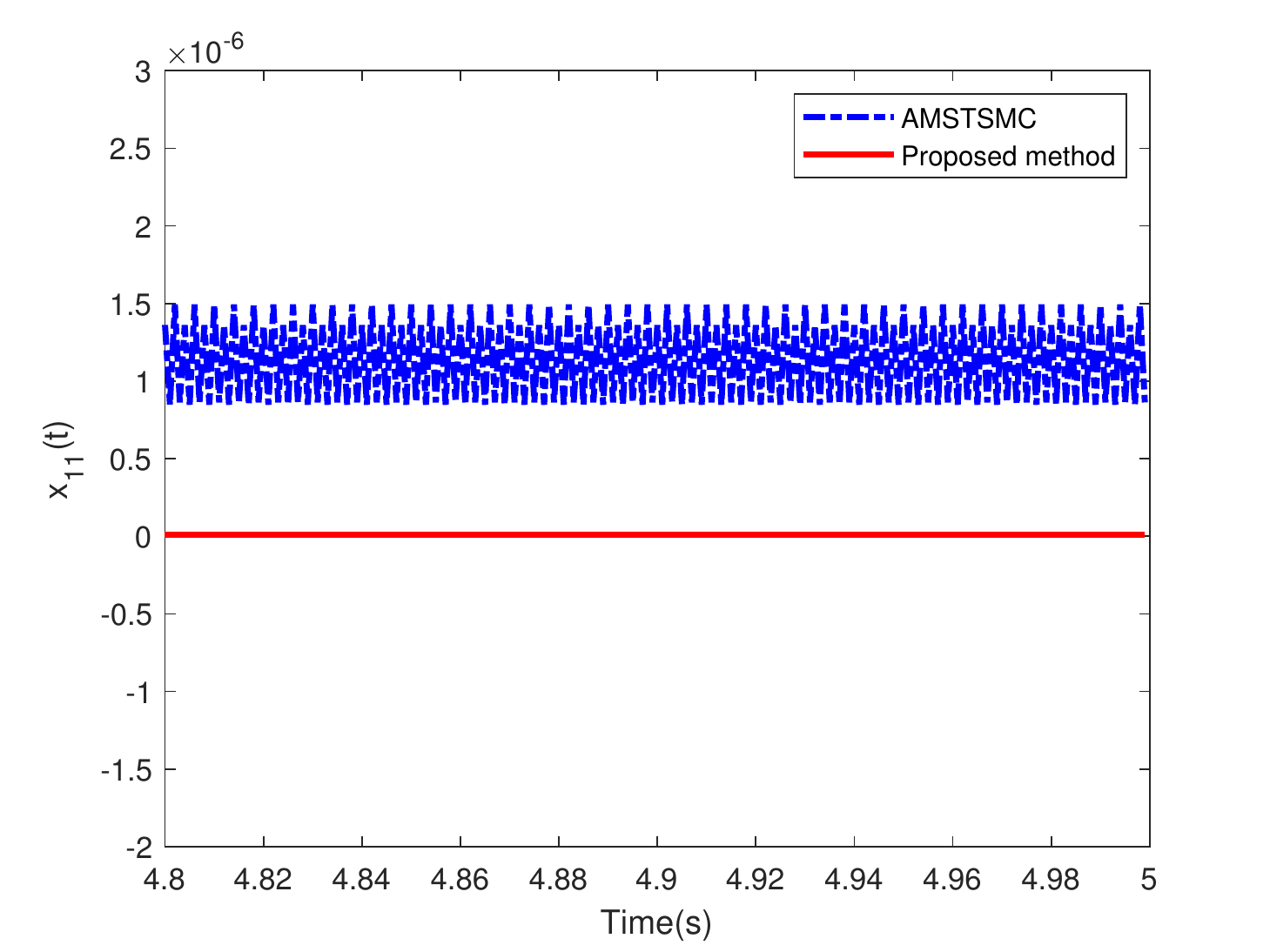}}
	\subfigure[The curves of ${x_{12}}$ with constant disturbance]{
	\includegraphics[width=0.2\textwidth]{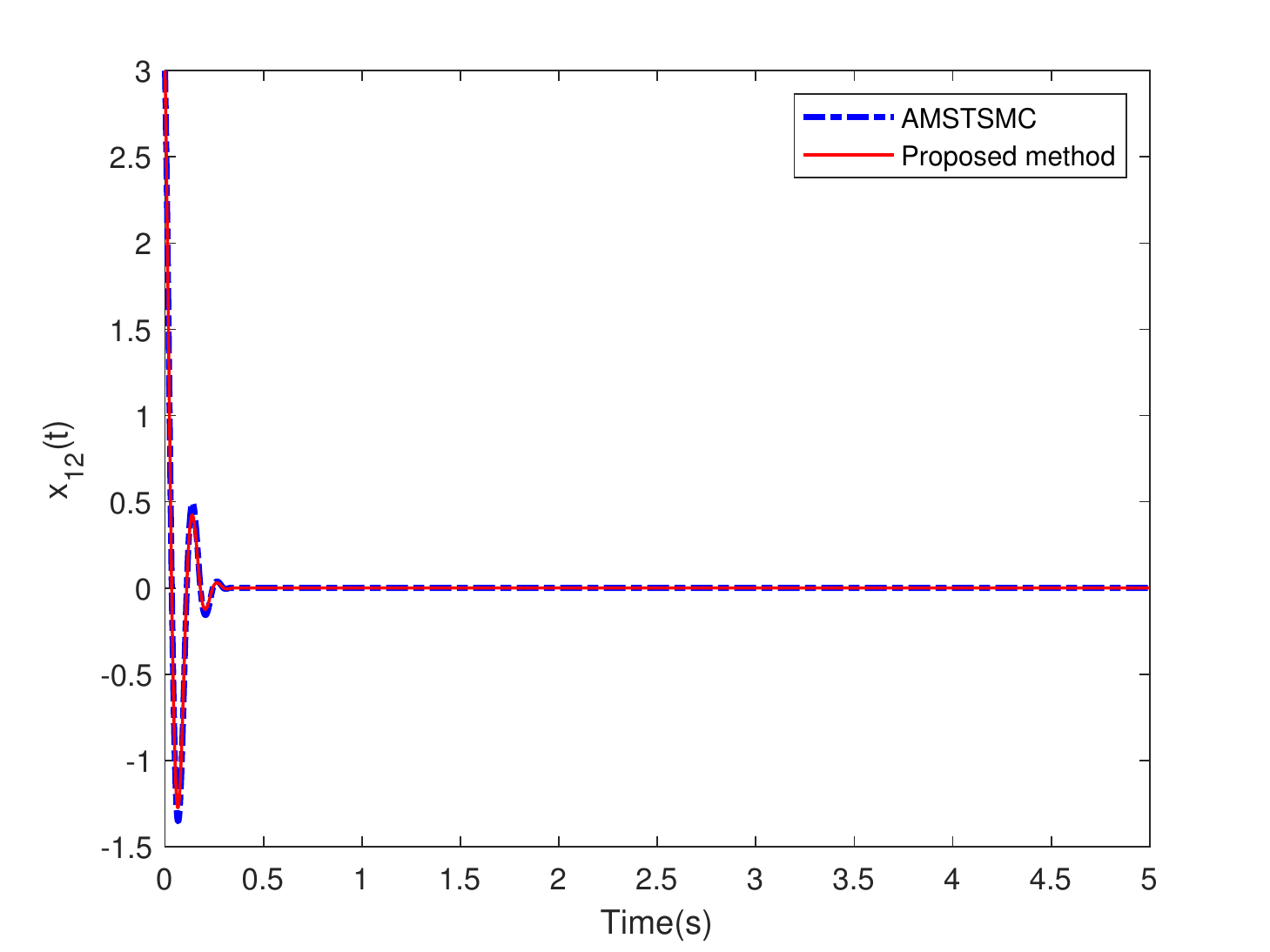}}
	\subfigure[Local magnification of ${x_{12}}$]{
	\includegraphics[width=0.2\textwidth]{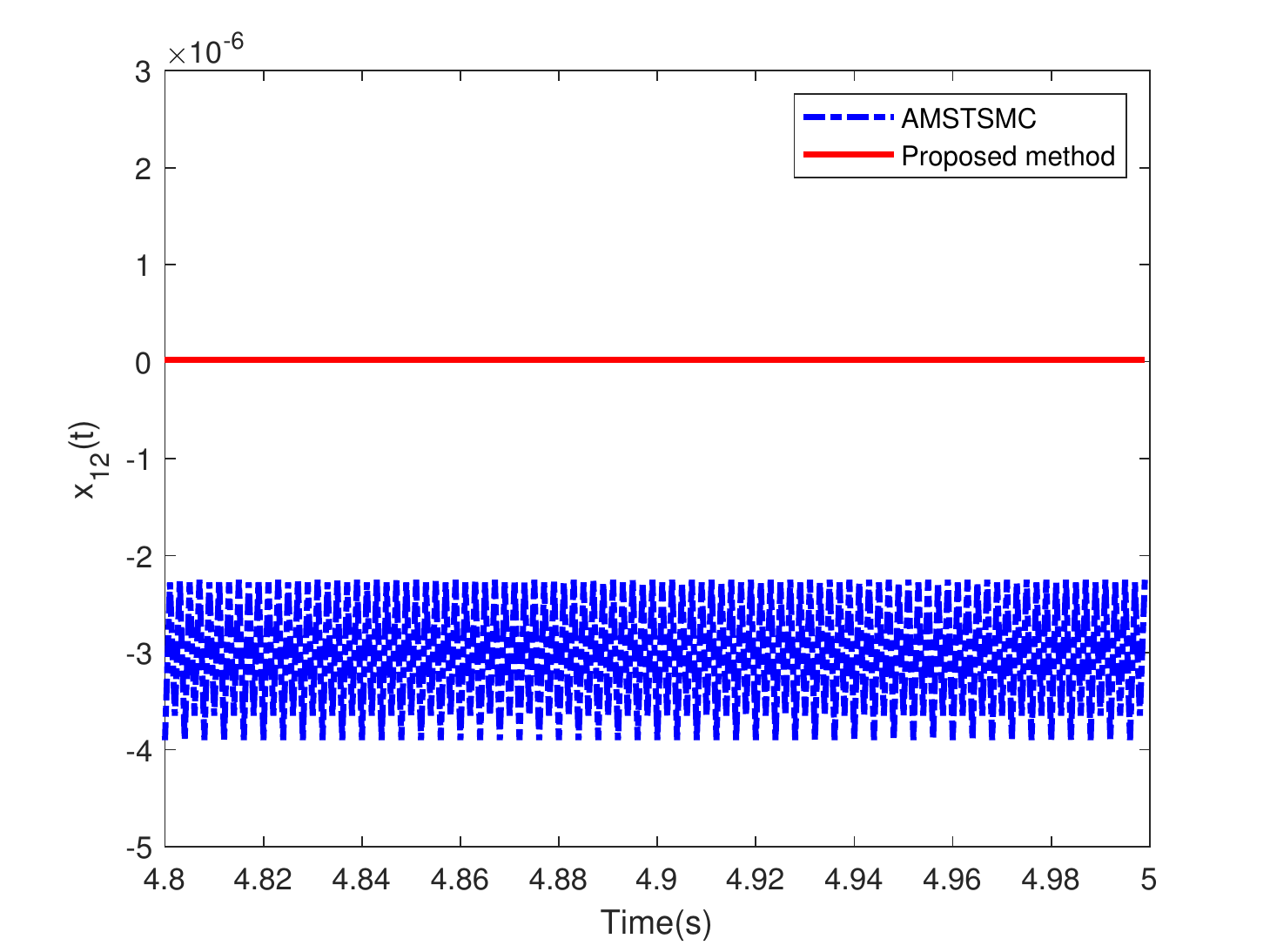}}
    \subfigure[The curves of ${x_{13}}$ with constant disturbance]{
	\includegraphics[width=0.2\textwidth]{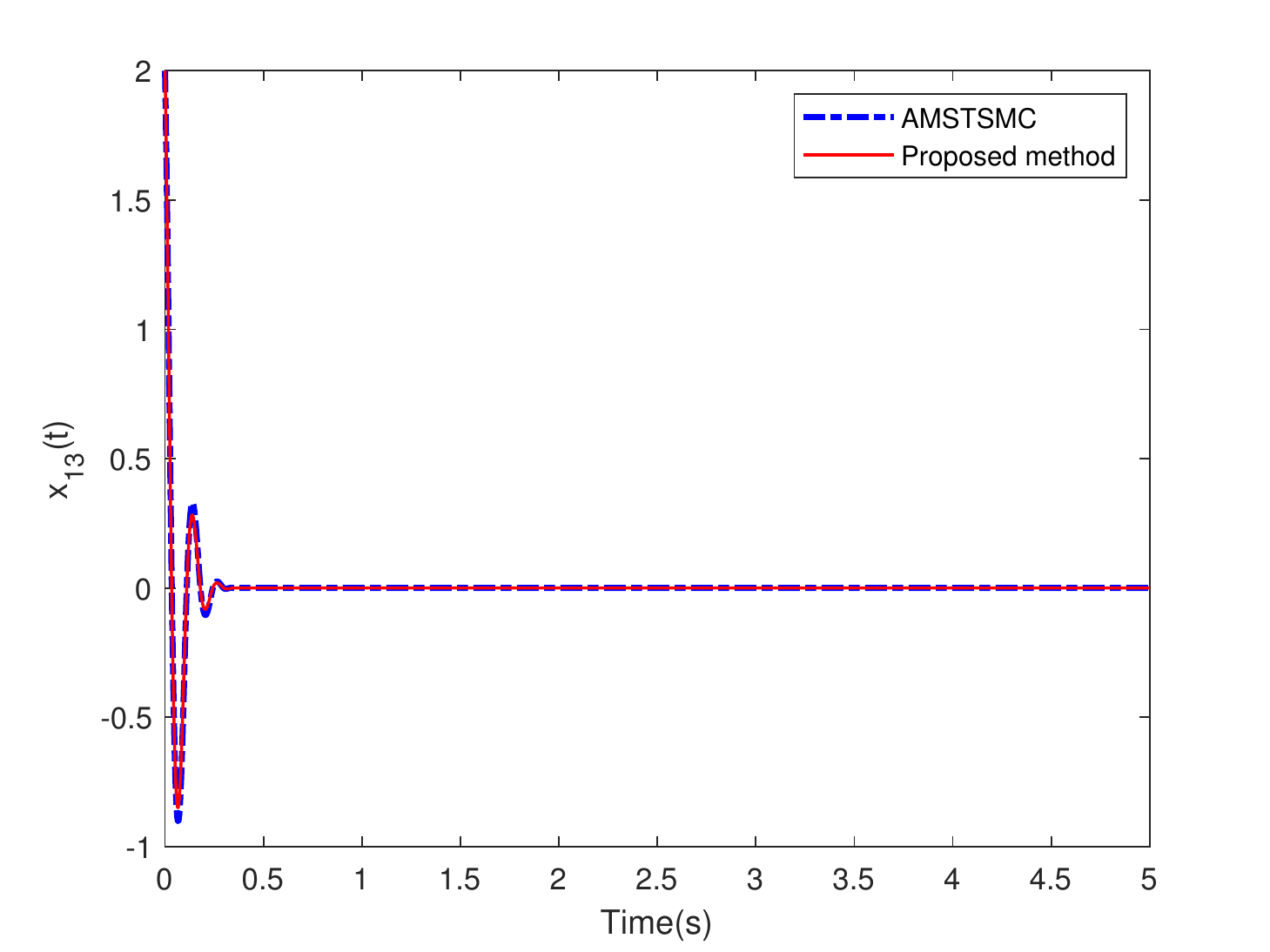}}
	\subfigure[Local magnification of ${x_{13}}$]{
	\includegraphics[width=0.2\textwidth]{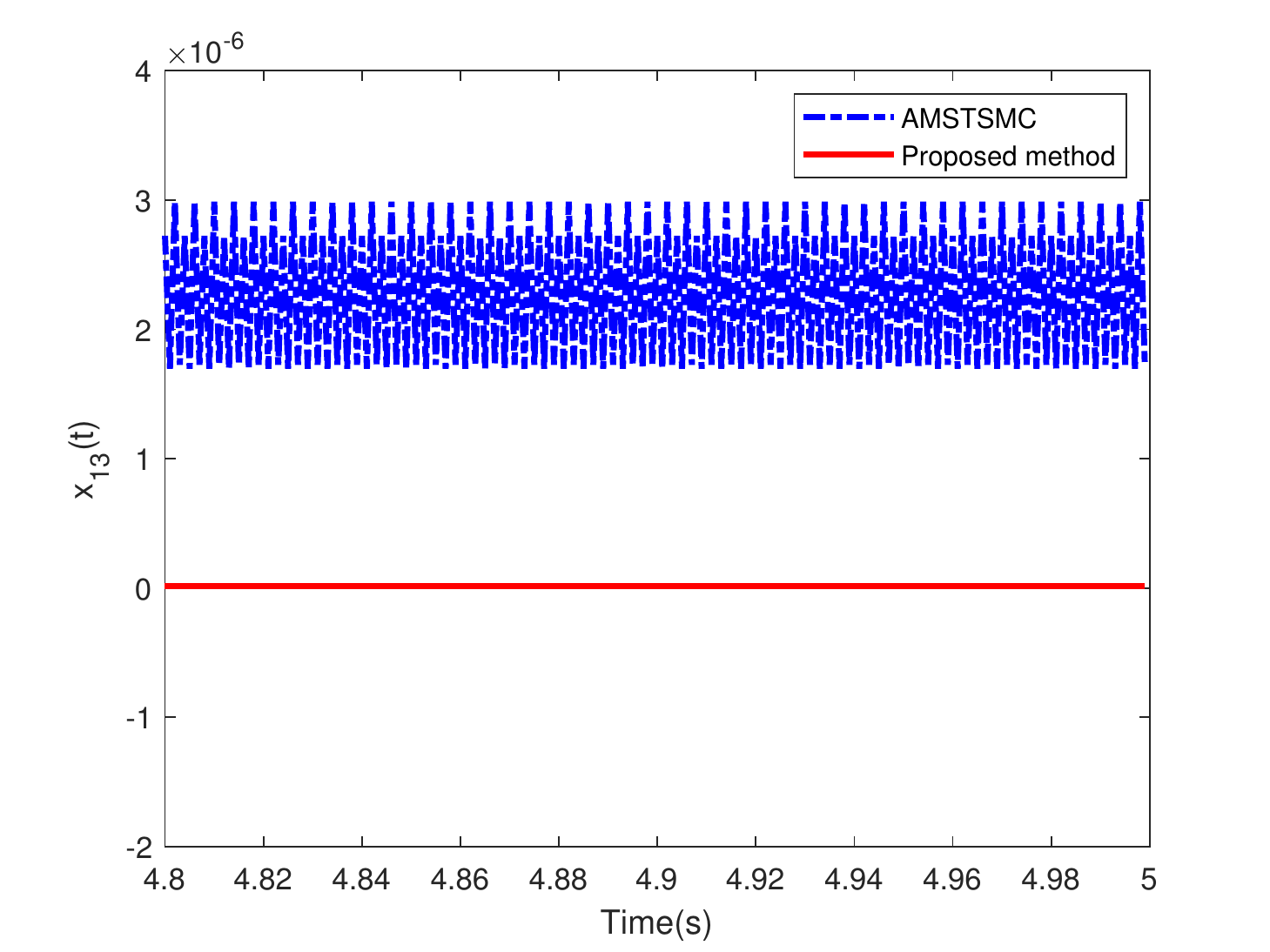}}
	\caption{Results of experiment \uppercase\expandafter{\romannumeral1}}
\end{figure}

The results of experiment \uppercase\expandafter{\romannumeral1} are given in Fig.1: (a)-(f). Fig. 1: (a), (c) and (e) illustrate the responses of state variables by utilizing AMSTSMC and the proposed method, respectively. Fig. 1: (b), (d) and (f) present the local magnification of corresponding state variable to show the steady-state response more clearly. Fig. 1: (a) shows that the state variable converges to the origin in finite time by using the proposed control method as fast as the method in [9], which means the newly proposed method also possesses the characteristic of fast finite-time convergence. In addition, Fig. 1: (b) demonstrates that the proposed method can diminish the chattering existing in the AMSTSMC method. A similar conclusion can be drawn from the Fig. 1: (c), (d) and (e), (f).

%4-B
\subsection{Experiment \uppercase\expandafter{\romannumeral2}: Comparison with time-varying disturbance}
In experiment \uppercase\expandafter{\romannumeral2}, the time-varying disturbance is set as ${d_1} = {\left[0.1\sin (t)\quad 0.2\cos (4t)\quad 0.2\cos (2t) \right]^T}$. The controllers and parameters are selected the same as those in the experiment \uppercase\expandafter{\romannumeral1}.
\begin{figure}
	\centering
	\subfigure[The curves of ${x_{11}}$ with time-varying disturbance]{
	\includegraphics[width=0.2\textwidth]{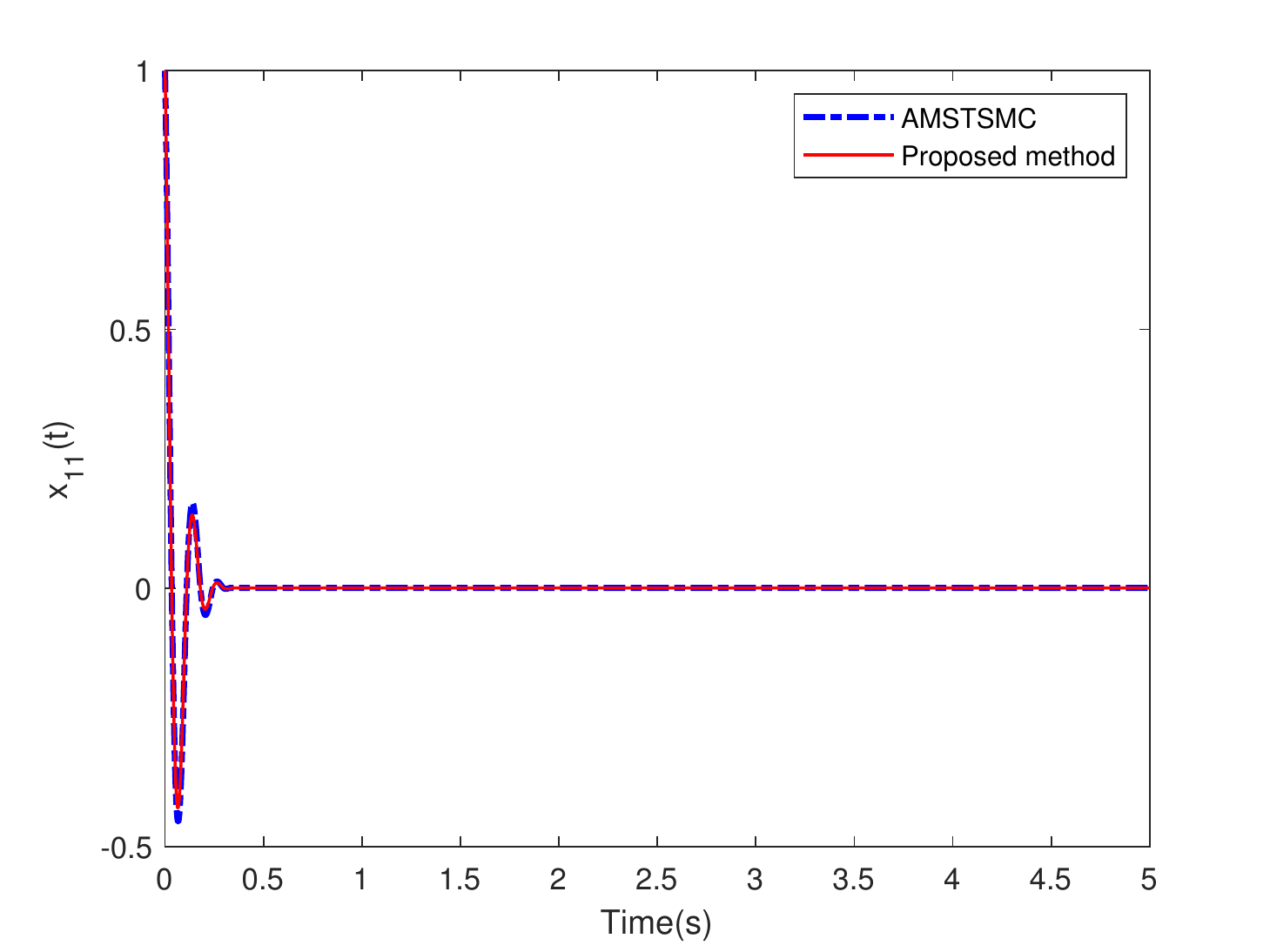}}
	\subfigure[Local magnification of ${x_{11}}$]{
	\includegraphics[width=0.2\textwidth]{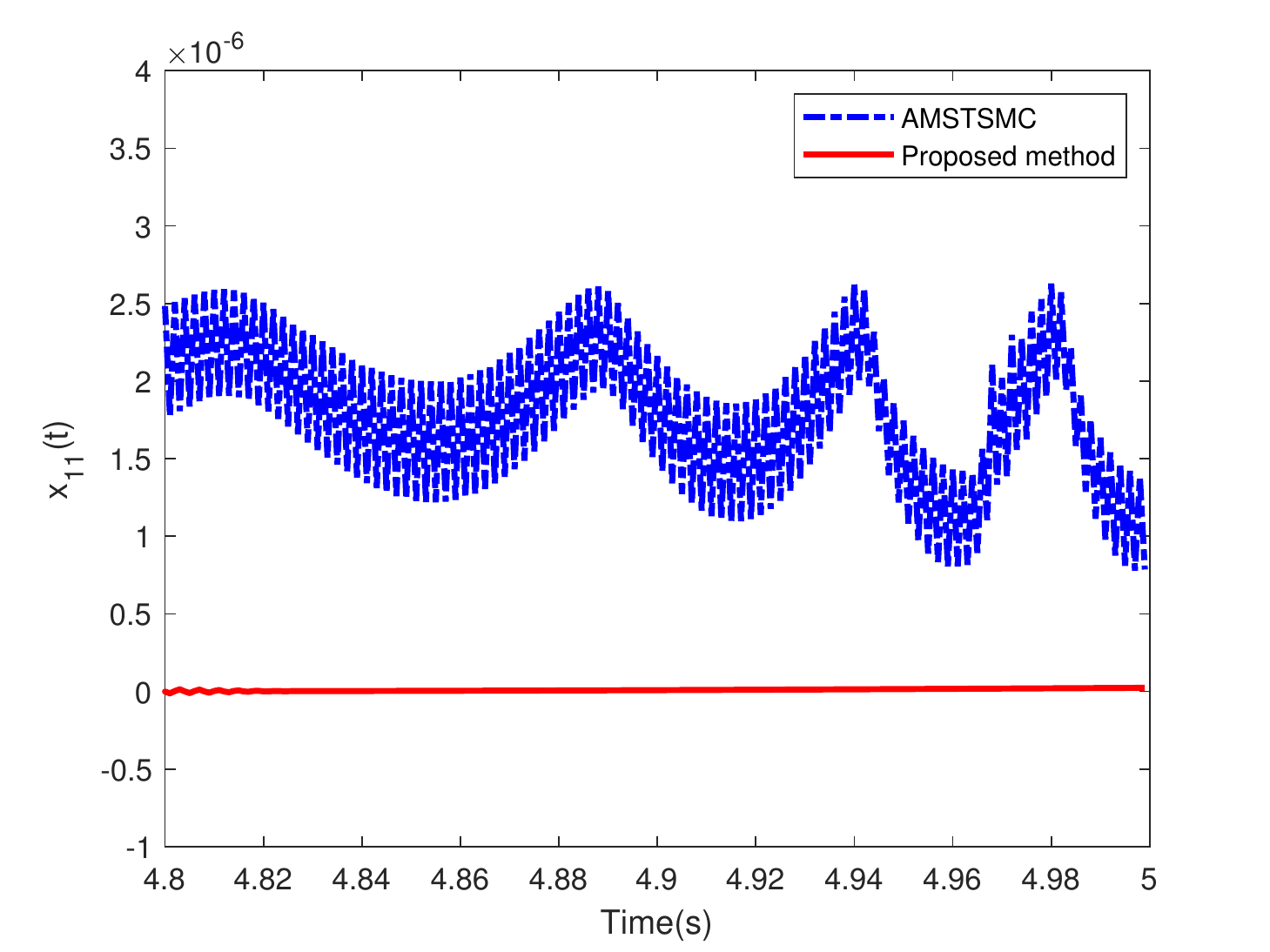}}
	\subfigure[The curves of ${x_{12}}$ with time-varying disturbance]{
	\includegraphics[width=0.2\textwidth]{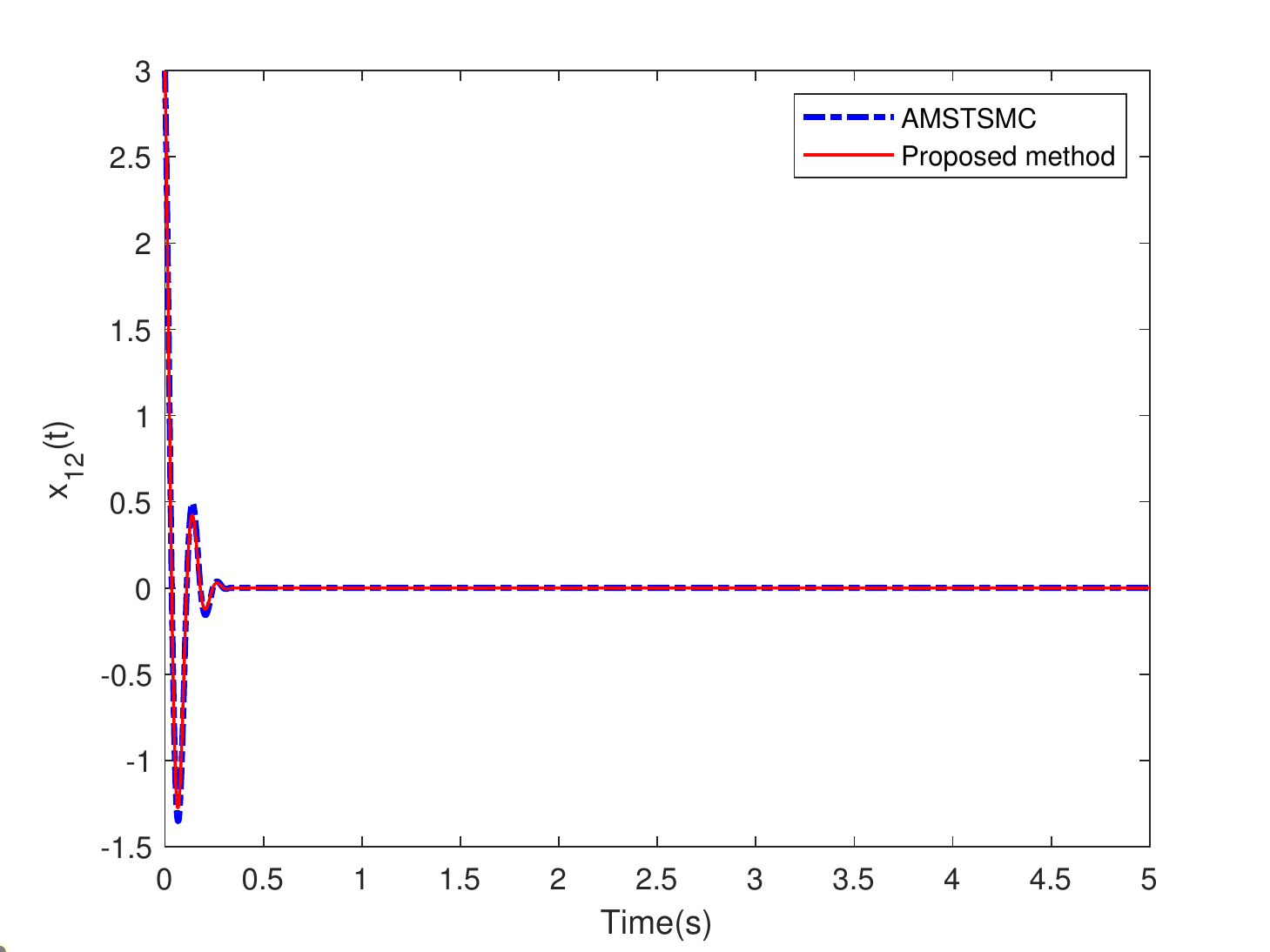}}
	\subfigure[Local magnification of ${x_{12}}$]{
	\includegraphics[width=0.2\textwidth]{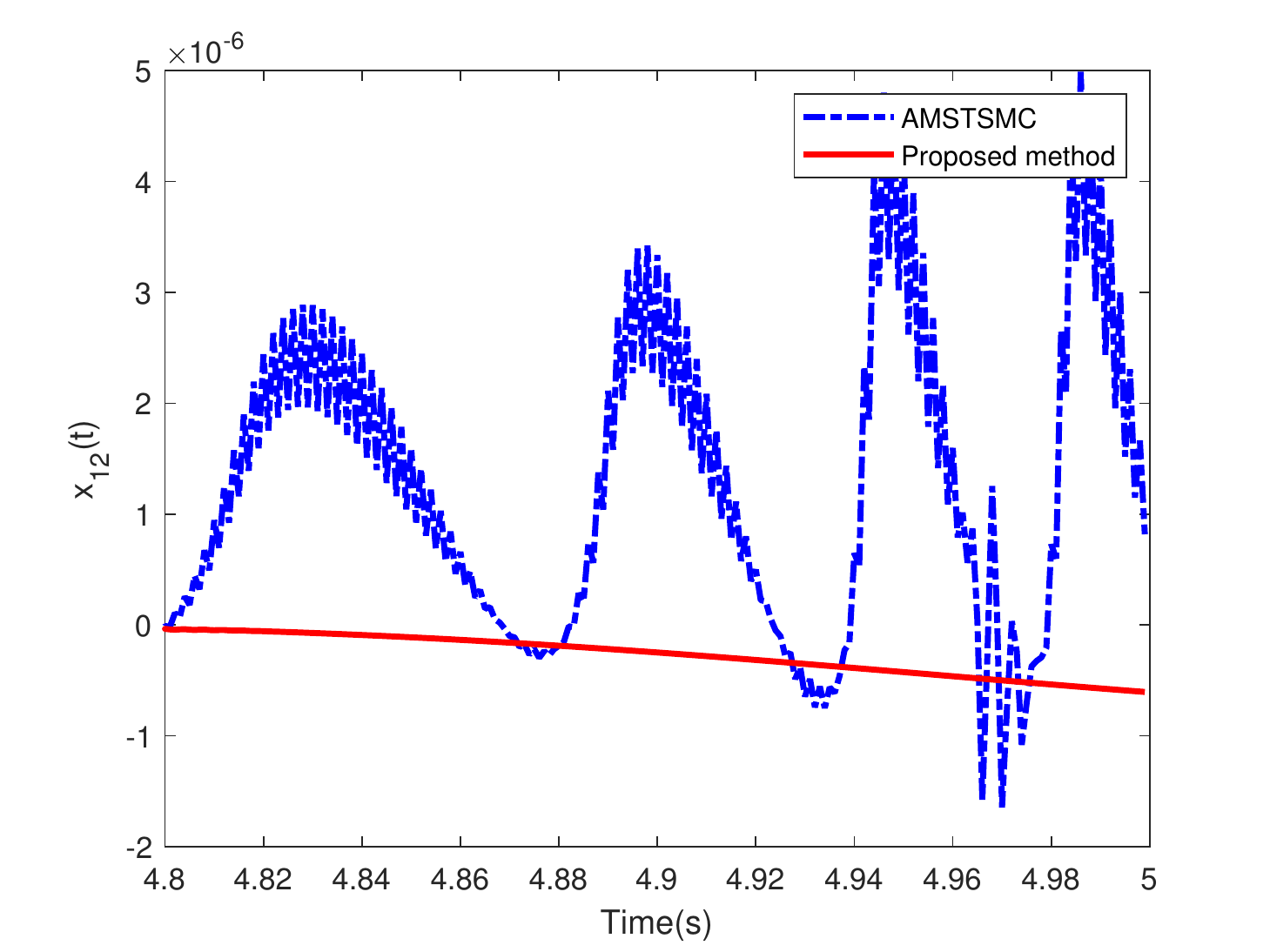}}
    \subfigure[The curves of ${x_{13}}$ with time-varying disturbance]{
	\includegraphics[width=0.2\textwidth]{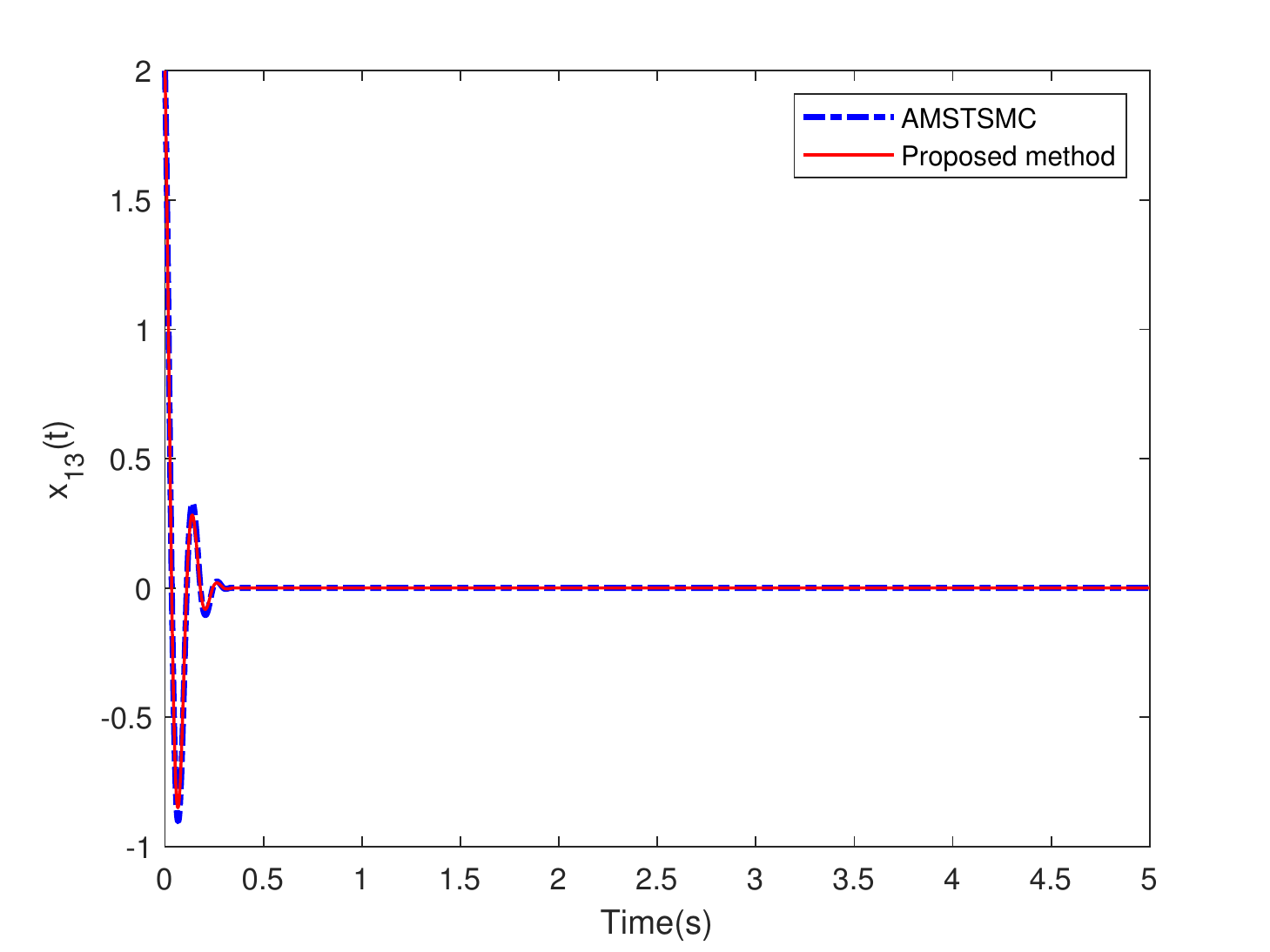}}
	\subfigure[Local magnification of ${x_{13}}$]{
	\includegraphics[width=0.2\textwidth]{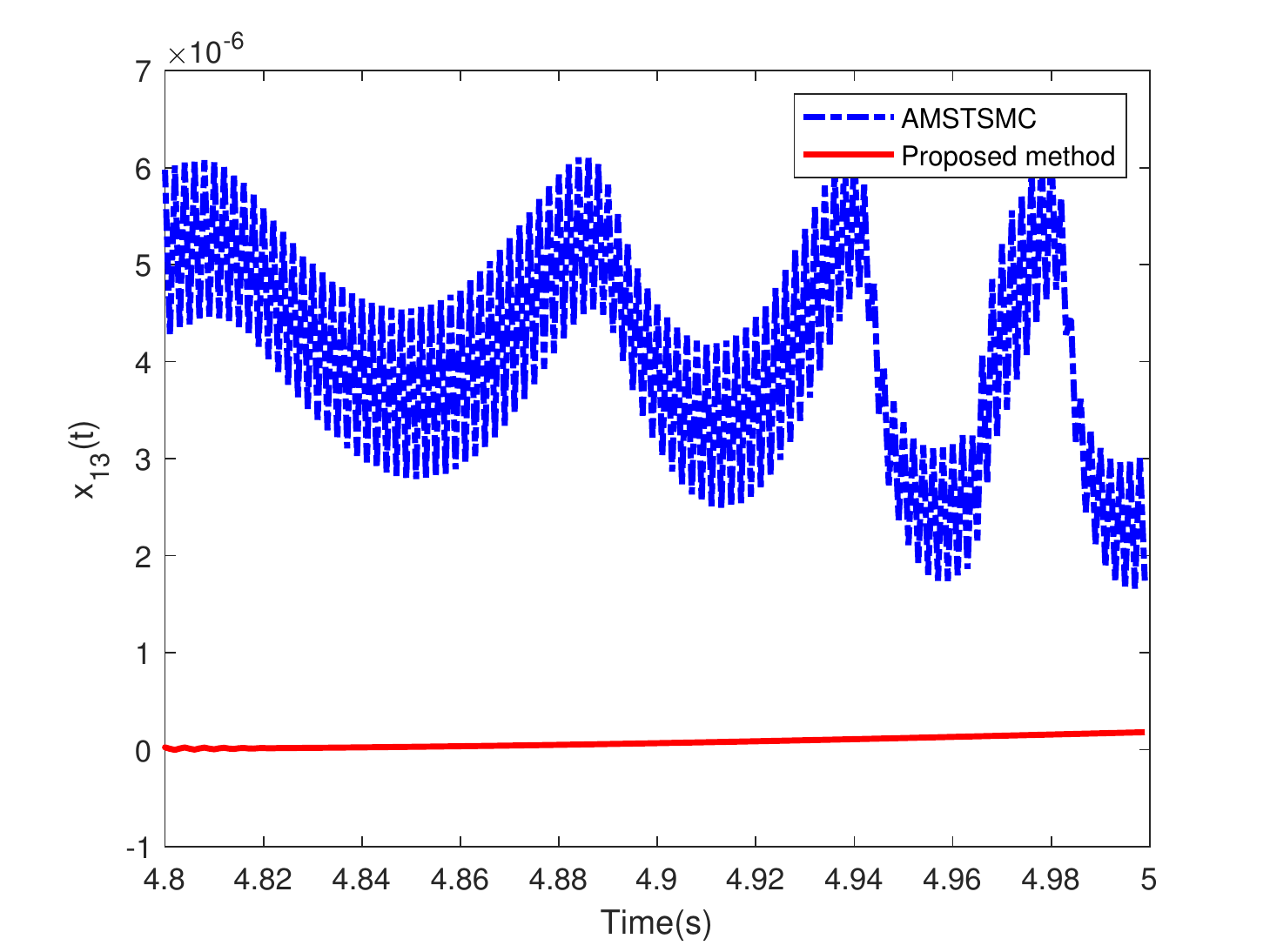}}
	\caption{Results of experiment \uppercase\expandafter{\romannumeral2}}
\end{figure}

The results of experiment \uppercase\expandafter{\romannumeral2} are shown in Fig. 2: (a)-(f). Fig. 2: (a), (c) and (e) demonstrate the responses of state variables by using AMSTSMC and the proposed method, respectively. Fig. 2: (b), (d) and (f) present the local magnification of corresponding state variable to show the steady-state response more clearly. From Fig. 2: (a) and (b), one can see that the proposed method can guarantee the fast finite-time uniformly ultimately boundedness coupled with enormous chattering suppression. A similar conclusion can be drawn from the Fig. 2: (c), (d) and (e), (f).

%4-C
\subsection{Experiment \uppercase\expandafter{\romannumeral3}: Observer comparison}
Based on \emph{Proposition 1}, the observer for estimating the disturbance ${d_1}$ is designed as follows
\begin{equation}\small
\begin{aligned}
{\hat d_1} = {L_{d1}}(t){\kern 1pt} \frac{{{e_1}}}{{{{\left\| {{e_1}} \right\|}^{{\textstyle{1 \over m}}}}}} &+ {L_{d2}}(t){e_1}\\ 
&+ \int_0^t {\left[ {{L_{d3}}(t)\frac{{{e_1}}}{{{{\left\| {{e_1}} \right\|}^{{\textstyle{2 \over m}}}}}} + {L_{d4}}(t){e_1}} \right]d\tau }  
\end{aligned}
\end{equation}
where ${e_1} = {z_1} - {x_1}$, ${\dot z_1} = u + {\hat d_1}$. The adaptive gains ${L_{di}}(t)(i=1,2,3,4)$ are formulated the same as (15) where the parameters are set as: $m = 3,{k_1} = 2,{k_2} = 2.5,{k_3} = 4,{k_4} = 30,\kappa  = 10$. 

The purpose of the experiment \uppercase\expandafter{\romannumeral3} is to compare the performance of the proposed observer with that of the adaptive multivariable disturbance observer (AMDO) in [9], which adopts the same parameters except that $m$ is set as $2$. To facilitate the comparison, the disturbance is set as ${d_1} = {\left[\sin (t)\quad 2\cos (4t)\quad 2\cos (2t) \right]^T}$. 
\begin{figure}
	\centering
	\subfigure[Disturbance estimation error curves of $d_{11}$]{
	\includegraphics[width=0.2\textwidth]{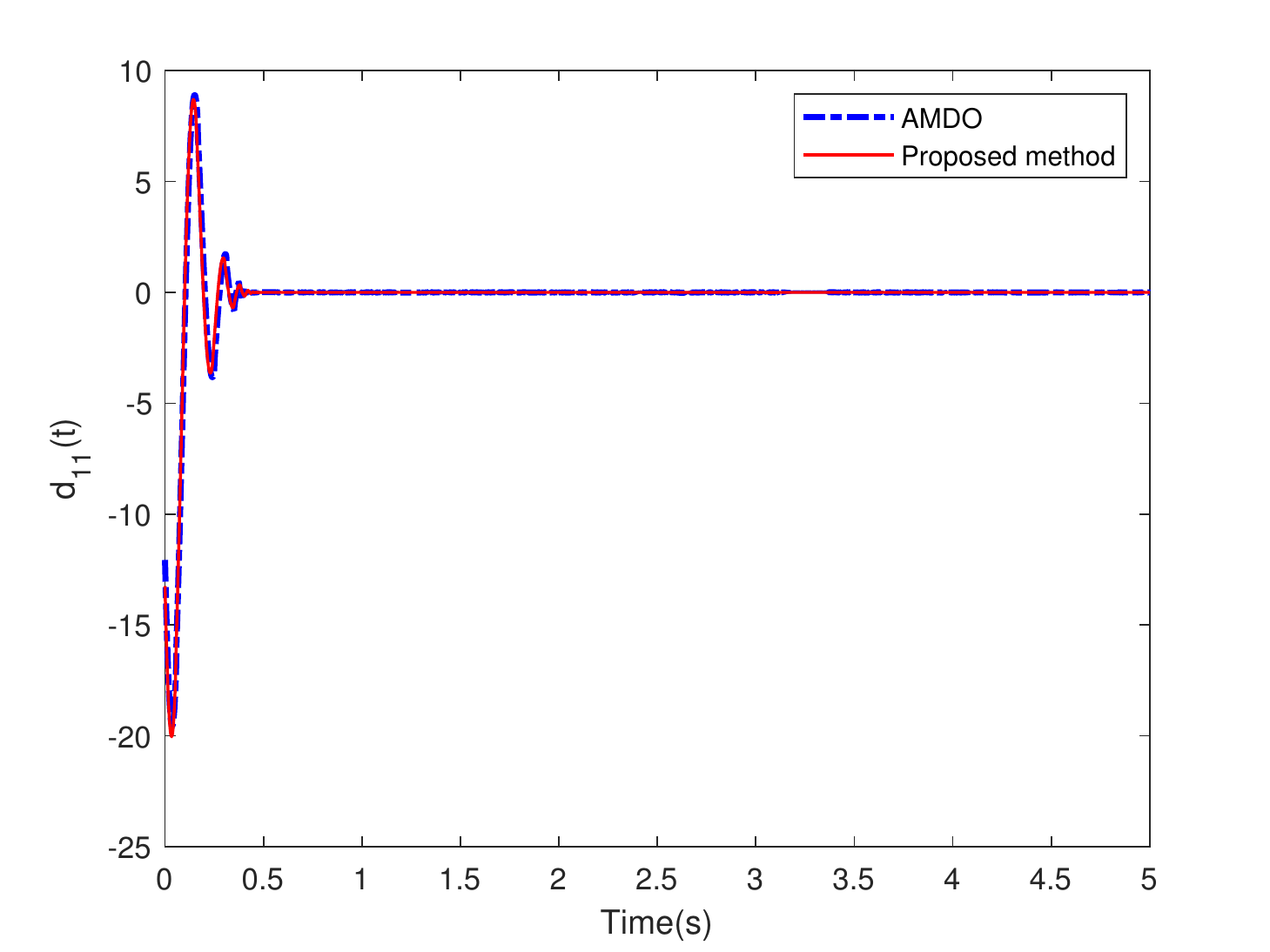}}
	\subfigure[Local magnification of disturbance estimation error]{
	\includegraphics[width=0.2\textwidth]{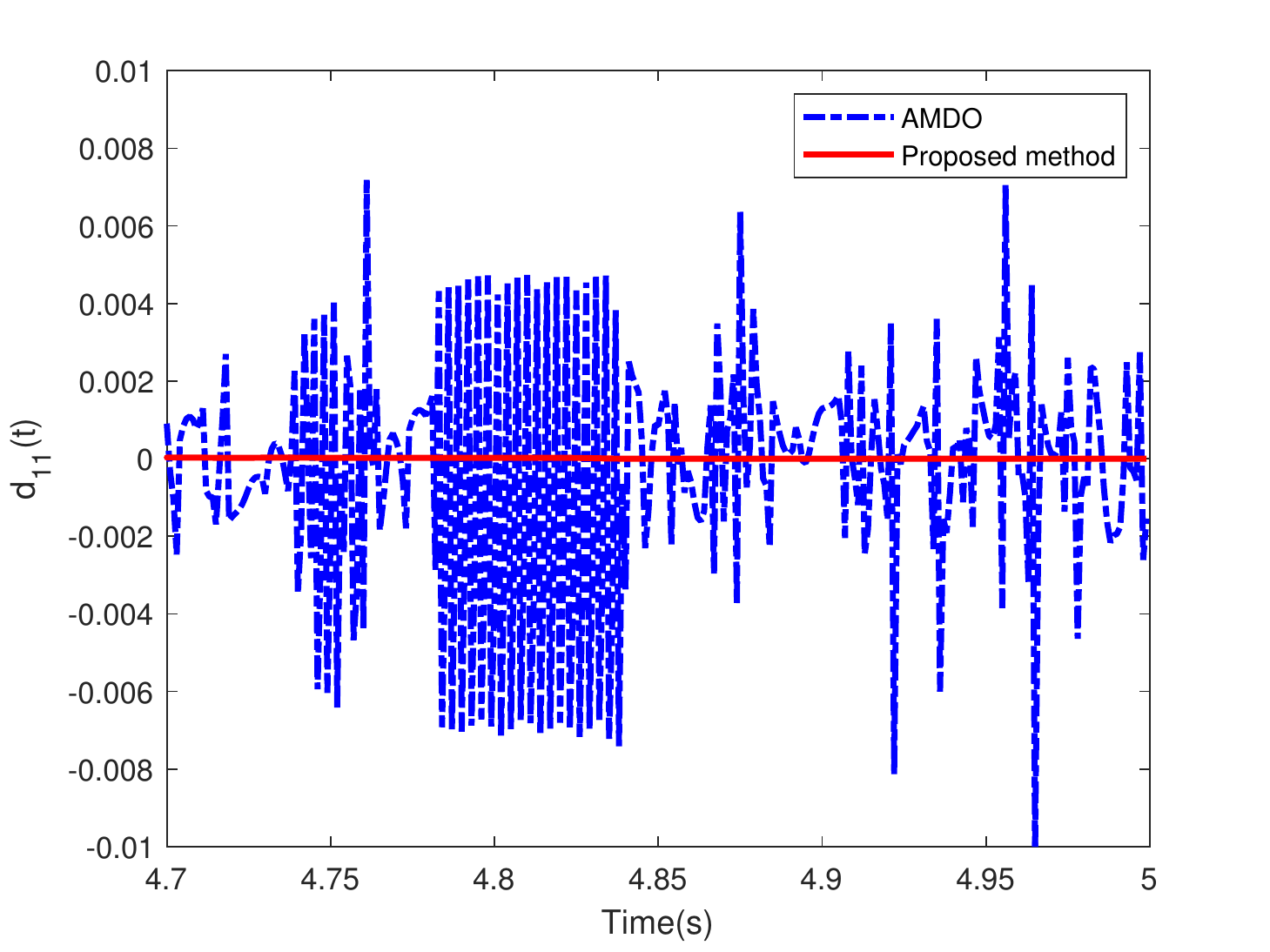}}
    \subfigure[Disturbance estimation error curves of $d_{12}$]{
	\includegraphics[width=0.2\textwidth]{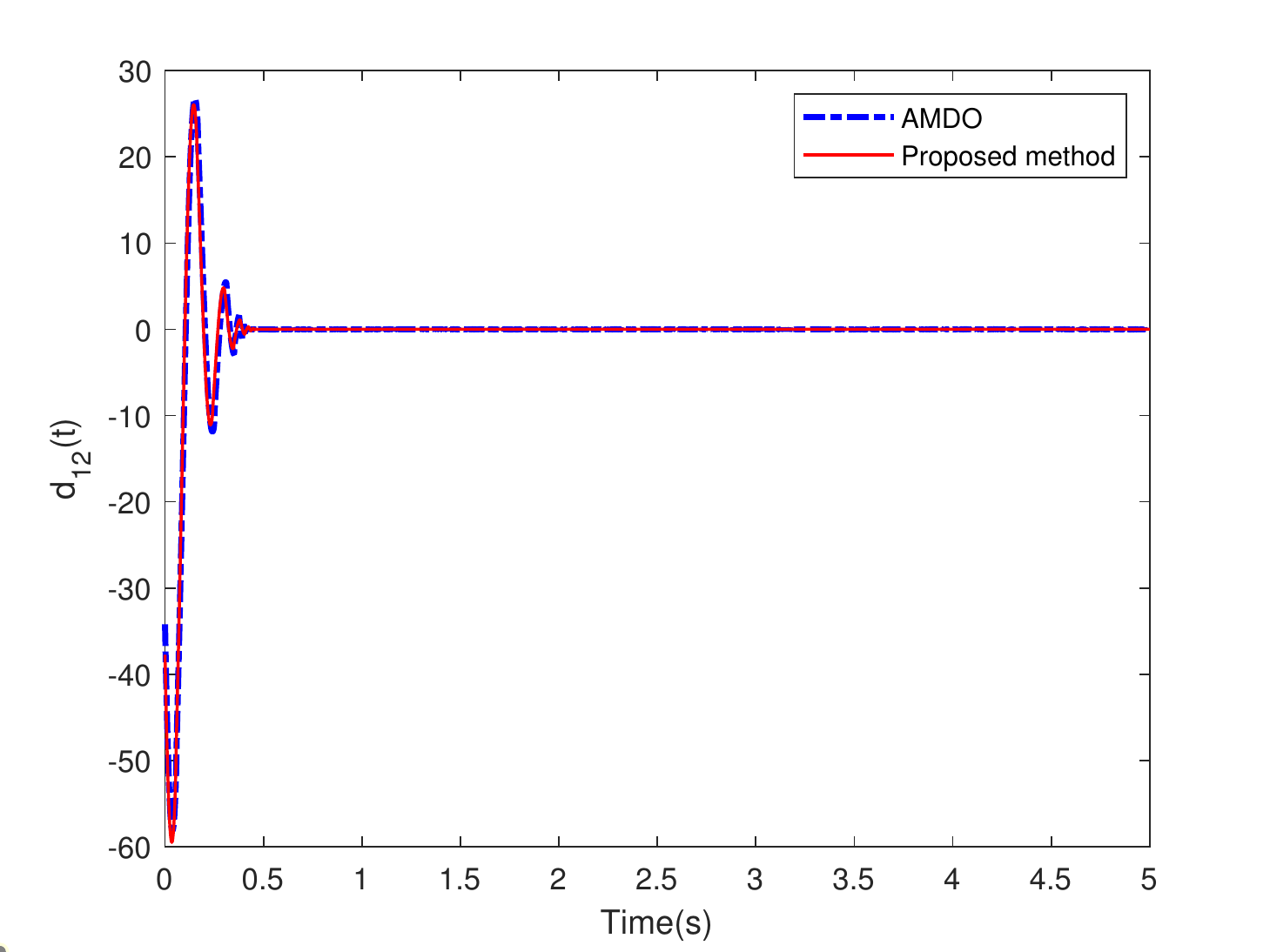}}
	\subfigure[Local magnification of disturbance estimation error]{
	\includegraphics[width=0.2\textwidth]{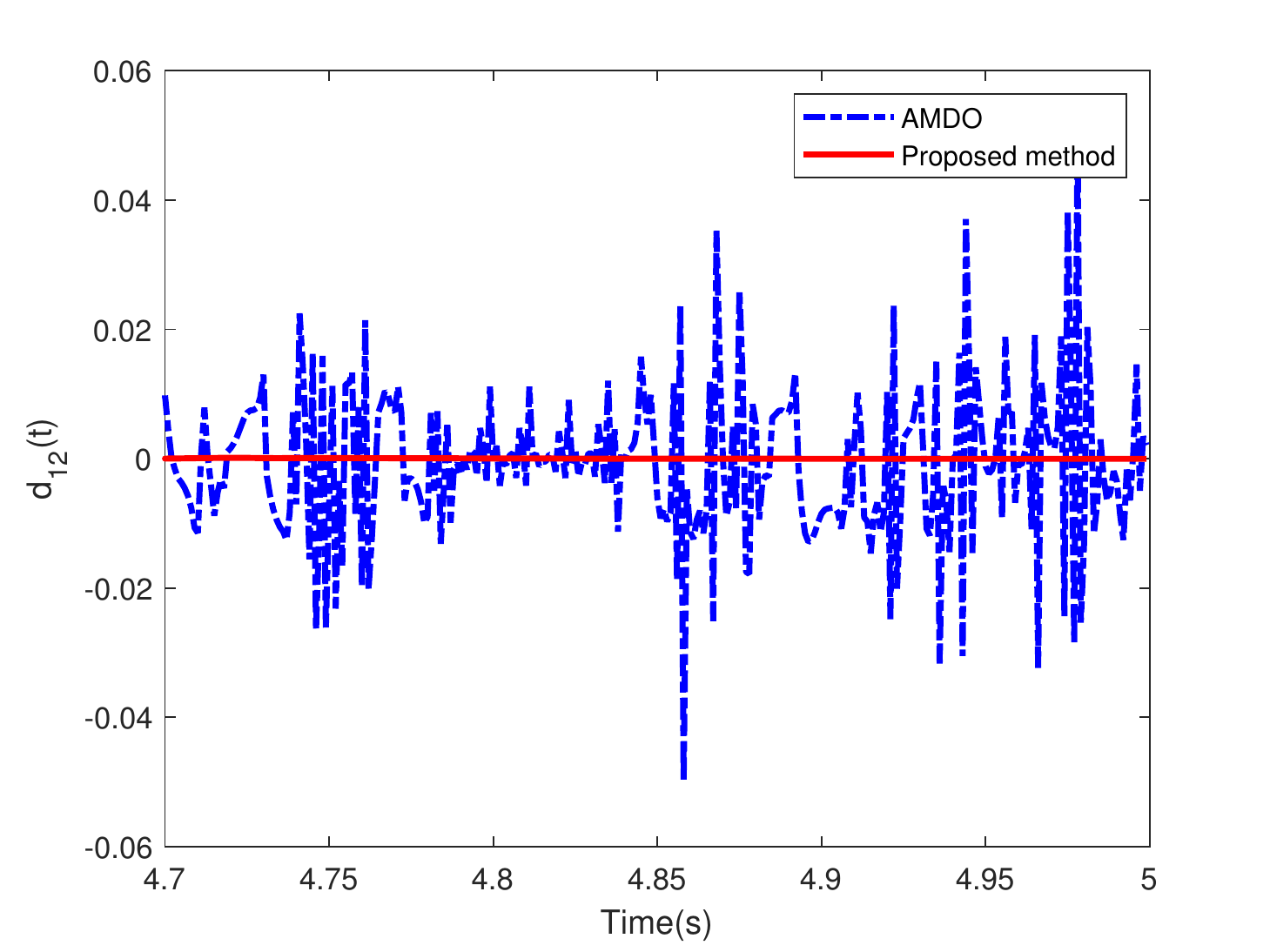}}
    \subfigure[Disturbance estimation error curves of $d_{13}$]{
	\includegraphics[width=0.2\textwidth]{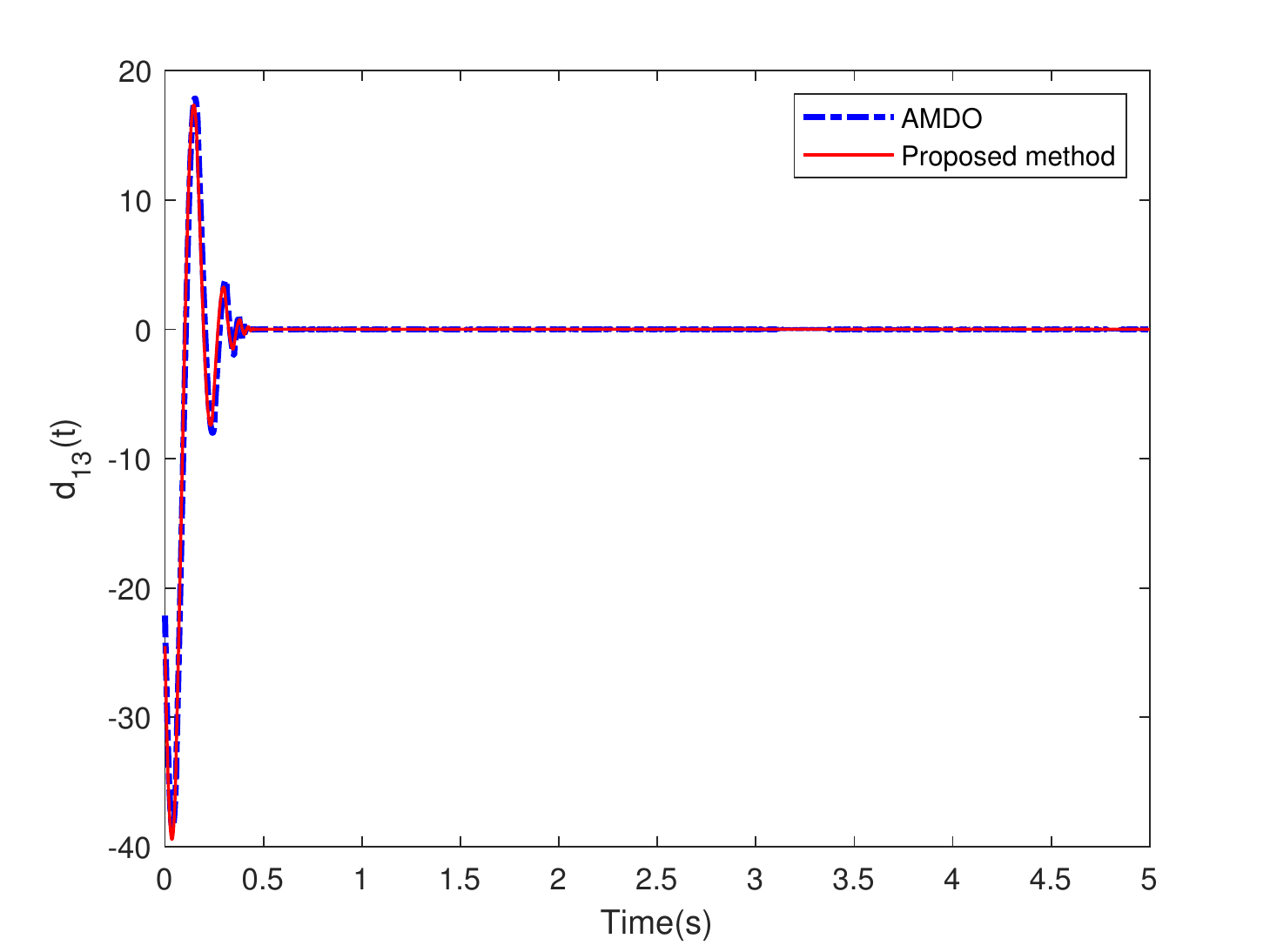}}
	\subfigure[Local magnification of disturbance estimation error]{
	\includegraphics[width=0.2\textwidth]{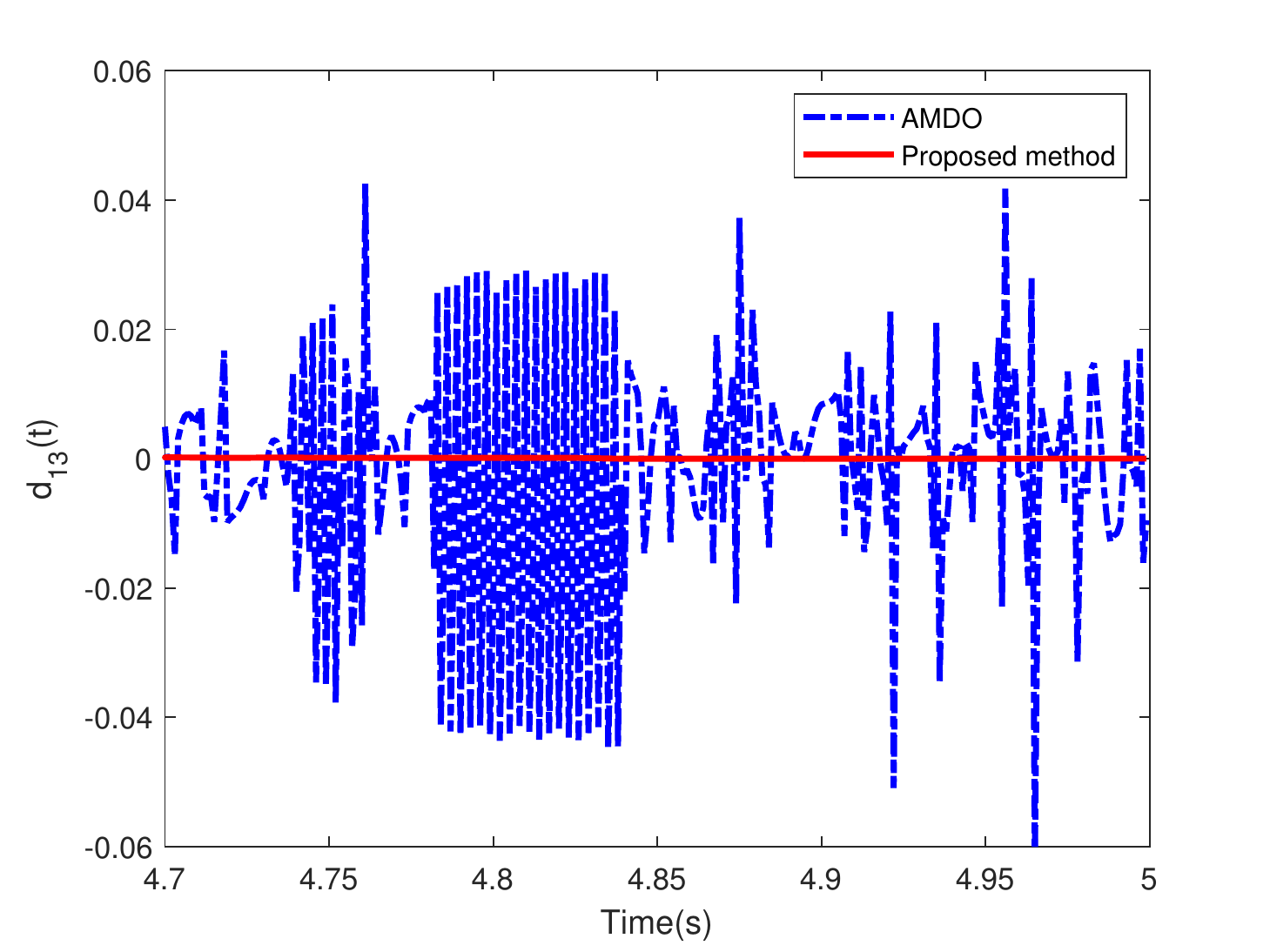}}
	\caption{Results of experiment \uppercase\expandafter{\romannumeral3}}
\end{figure}

The results of experiment \uppercase\expandafter{\romannumeral3} are given in Fig. 3: (a)-(f). Fig. 3: (a), (c) and (e) indicate the responses of observer estimation error components by using AMDO and the proposed observer. Fig. 3: (b), (d) and (f) present the local magnification to show the responses more clearly. From Fig. 3: (a) and (b), one can see that the estimation error can fast converge to a region of the origin in finite time by using the proposed observer. In addition, the proposed observer can effectively alleviate chattering effect existing in the AMDO and provide smoother output for the disturbance estimation. A similar conclusion can be drawn from the Fig. 3: (c), (d) and (e), (f).

\section{Conclusion}
In this paper, a novel adaptive multivariable smooth second-order sliding mode approach has been proposed. This newly proposed approach is utilized in the design of controller and observer for MIMO systems. According to the types of disturbances, the fast finite-time convergence and the fast finite-time uniformly ultimately boundedness of the systems are proved with the corresponding finite-time Lyapunov stability theory. The comparative numerical simulations are performed to demonstrate the effectiveness and superiority of the proposed approach with fast finite-time convergence, adaptation to disturbances, and chattering suppression for the MIMO system.

% References
\bibliographystyle{Bibliography/IEEEtranTIE}
\bibliography{Bibliography/IEEEabrv,Bibliography/myRef}\ %IEEEabrv instead of IEEEfull

\begin{thebibliography}{10}
\providecommand{\url}[1]{#1}
\csname url@samestyle\endcsname
\providecommand{\newblock}{\relax}
\providecommand{\bibinfo}[2]{#2}
\providecommand{\BIBentrySTDinterwordspacing}{\spaceskip=0pt\relax}
\providecommand{\BIBentryALTinterwordstretchfactor}{4}
\providecommand{\BIBentryALTinterwordspacing}{\spaceskip=\fontdimen2\font plus
\BIBentryALTinterwordstretchfactor\fontdimen3\font minus
  \fontdimen4\font\relax}
\providecommand{\BIBforeignlanguage}[2]{{%
\expandafter\ifx\csname l@#1\endcsname\relax
\typeout{** WARNING: IEEEtran.bst: No hyphenation pattern has been}%
\typeout{** loaded for the language `#1'. Using the pattern for}%
\typeout{** the default language instead.}%
\else
\language=\csname l@#1\endcsname
\fi
#2}}
\providecommand{\BIBdecl}{\relax}
\BIBdecl

\bibitem{Levant2003}
Levant and Arie, ``Higher-order sliding modes, differentiation and
  output-feedback control,'' \emph{International Journal of Control}, vol.~76,
  no. 9-10, pp. 924--941, 2003.

\bibitem{2007Smooth}
Y.~B. Shtessel, I.~A. Shkolnikov, and A.~Levant, ``Smooth second-order sliding
  modes: Missile guidance application,'' \emph{Automatica}, vol.~43, no.~8, pp.
  1470--1476, 2007.

\bibitem{Moreno2008}
J.~A. {Moreno} and M.~{Osorio}, ``A {Lyapunov} approach to second-order sliding
  mode controllers and observers,'' in \emph{2008 47th IEEE Conference on
  Decision and Control}, pp. 2856--2861, 2008.

\bibitem{Shtessel2012}
Y.~Shtessel, M.~Taleb, and F.~Plestan, ``A novel adaptive-gain supertwisting
  sliding mode controller: Methodology and application,'' \emph{Automatica},
  vol.~48, no.~5, pp. 759 -- 769, 2012.

\bibitem{ASTC2015}
S.~{Laghrouche}, J.~{Liu}, F.~S. {Ahmed}, M.~{Harmouche}, and M.~{Wack},
  ``Adaptive second-order sliding mode observer-based fault reconstruction for
  pem fuel cell air-feed system,'' \emph{IEEE Transactions on Control Systems
  Technology}, vol.~23, no.~3, pp. 1098--1109, 2015.

\bibitem{Jiang2018}
Q.~Hu and B.~Jiang, ``Continuous finite-time attitude control for rigid
  spacecraft based on angular velocity observer,'' \emph{IEEE Transactions on
  Aerospace and Electronic Systems}, vol.~54, no.~3, pp. 1082--1092, 2018.

\bibitem{2020arXiv}
X.~{Wang}, Z.~{Li}, Z.~{He}, and H.~{Gao}, ``Adaptive fast smooth second-order
  sliding mode control for attitude tracking of a 3-dof helicopter,''
  \emph{arXiv e-prints}, p. arXiv:2008.10817, Aug. 2020.

\bibitem{2014Multi}
I.~Nagesh and C.~Edwards, ``A multivariable super-twisting sliding mode
  approach,'' \emph{Automatica}, vol.~50, no.~3, pp. 984--988, 2014.

\bibitem{2015Adaptive_Multi}
B.~Tian, L.~Yin, and H.~Wang, ``Finite-time reentry attitude control based on
  adaptive multivariable disturbance compensation,'' \emph{IEEE Transactions on
  Industrial Electronics}, vol.~62, no.~9, pp. 5889--5898, 2015.

\bibitem{1999Filippov}
A.~F. Filippov, ``Differential equations with discontinuous righthand sides,''
  \emph{Journal of Mathematical Analysis \& Applications}, vol. 154, no.~2, pp.
  99--128, 1999.

\bibitem{2005lemma1}
S.~Yu, X.~Yu, B.~Shirinzadeh, and Z.~Man, ``Continuous finite-time control for
  robotic manipulators with terminal sliding mode,'' \emph{Automatica},
  vol.~41, no.~11, pp. 1957--1964, 2005.

\end{thebibliography}

\end{document}